\documentclass[a4paper,10pt]{amsart}
\usepackage[utf8]{inputenc}

\RequirePackage[mathcal,mathbf]{euler}

\usepackage{amsthm}
\usepackage{graphicx}
\usepackage[dvips]{color}
\usepackage{epsfig}
\usepackage{psfrag}
\usepackage[all]{xy}
\usepackage{amscd}
\usepackage{amssymb}
\usepackage{fullpage}
\usepackage{url}
\usepackage{xcolor}
\usepackage{supertabular}

\usepackage{tikz}
\def\dar[#1]{\ar@<2pt>[#1]\ar@<-2pt>[#1]}
\entrymodifiers={!!<0pt,0.7ex>+}

\newtheorem{Thm}{Theorem}[section]
\newtheorem{Prop}[Thm]{Proposition}
\newtheorem{Lem}[Thm]{Lemma}
\newtheorem{Cor}[Thm]{Corollary}

\theoremstyle{remark}
\newtheorem{Rem}[Thm]{Remark}

\theoremstyle{definition}

\newtheorem{Def}[Thm]{Definition}
\newtheorem{Exa}[Thm]{Example}

\DeclareMathOperator{\Der}{Der}

\DeclareMathOperator{\Ad}{Ad}
\DeclareMathOperator{\ad}{ad}
\DeclareMathOperator{\Diff}{Diff}
\DeclareMathOperator{\Aut}{Aut}
\DeclareMathOperator{\GL}{GL}
\DeclareMathOperator{\sll}{\mathfrak{sl}}
\DeclareMathOperator{\id}{id}

\DeclareMathOperator{\Hom}{Hom}
\DeclareMathOperator{\Lie}{Lie}

\DeclareMathOperator{\TAut}{TAut}
\DeclareMathOperator{\tder}{tder}
\DeclareMathOperator{\Fun}{Fun}
\DeclareMathOperator{\End}{End}
\let\O\undefined
\DeclareMathOperator{\O}{\mathcal{O}}

\newcommand{\eTAut}[1][n]{\Exp(\Omega^1(\L_#1))}

\newcommand{\esder}{\Omega^1_\text{cl}(\L_n)}

\DeclareMathOperator{\Exp}{\mathcal E \textit{xp}}
\DeclareMathOperator{\Map}{Map}
\DeclareMathOperator{\Ham}{\mathcal H am}
\DeclareMathOperator{\tr}{tr}

\newcommand{\g}{\mathfrak{g}}

\newcommand\kk{{\mathbf k}}
\newcommand{\vf}{\mathfrak{X}}

\renewcommand{\L}{\mathcal{L}}

\newcommand{\emptycomment}[1]{}

\author{Florian Naef}
\title{Poisson Brackets in Kontsevich's "Lie World"}
\date{}
\begin{document}

\maketitle

\begin{abstract}
\noindent{\bf Abstract.} In this note the notion of Poisson brackets in Kontsevich's "Lie World" is developed. These brackets can be thought of as "universally" defined classical Poisson structures, namely formal expressions only involving the structure maps of a quadratic Lie algebra. We prove a uniqueness statement about these Poisson brackets with a given moment map. As an application we get formulae for the linearization of the quasi-Poisson structure of the moduli space of flat connections on a punctured sphere, and thereby identify their symplectic leaves with the reduction of coadjoint orbits. Equivalently, we get linearizations for the Goldman double Poisson bracket, our definition of Poisson brackets coincides with that of Van Den Bergh \cite{vdb} in this case.  This can furthermore be interpreted as giving a monoidal equivalence between Hamiltonian quasi-Poisson spaces and Hamiltonian spaces.
\end{abstract}

\emptycomment{
\section{Notation}
Let throughout $g \cong g^*$ denote a quadratic Lie algebra, with non-degenerate pairing denoted by $\left<\cdot , \cdot \right>$.
Everything is formal, except where it isn't.
}

\section{Introduction}
The motivation of this note was originally to give another proof of the result Theorem 6.6 in \cite{Jef}. Theorem 6.6 states that the moduli space of flat $\g$-connections on a surface of genus 0 with prescibed monodromy around the punctures is symplectomorphic to the symplectic reduction of the product of coadjoint orbits, at least if the prescribed monodromies are sufficiently close to the identity. The procedure is as follows. We identify the relevant moduli space with a reduction of a Hamiltonian quasi-Poisson space whose underlying manifold is the product of a number of $G$'s. Using the exponential map we can furthermore pull the situation back to a product of $\g$'s. Summarizing, we get a quasi-Poisson structure on $\g \times \dots \times \g$ together with a moment map, which we with to compare to the standard Kostant-Kirillov-Souriau structure. Moreover, all those structures are defined by "formulae" only involving the lie bracket and the inner product of $\g$. A precise definition of this is given below. It turns out that for such universally defined Hamiltonian (quasi-)Poisson structures the moment map uniquely defines the bivector field and vice versa.

\section{Lie spaces}
We recall some definitions from \cite{K}. Let $\mathbf{Lie}$ denote the category of free complete graded (super-)Lie algebras, where morphisms are continuous Lie algebra morphisms. We define the category $\mathbf{LieSp}$ of (formal affine) Lie spaces as the opposite category of $\Lie$. This definition is very much in analogy with the equivalence of (commutative) affine schemes and $\mathbf{Ring}^\text{op}$. Much of the language that follows is motivated by this analogy. By definition, there is a canonical contravariant functor
\begin{align*}
    \mathcal{O}: \mathbf{LieSp} = \mathbf{Lie}^\text{op} &\longrightarrow \mathbf{Lie}.
\end{align*}
More concretely, a Lie space $\L$ is nothing but a (graded) Lie algebra, which we choose to call the coordinate Lie algebra of the space and denote it by $\mathcal{O}(\L)$. Morphisms between Lie spaces are maps of Lie algebras in the opposite direction. A choice of free homogenous generators of $\mathcal{O}(\L)$ shall be called a coordinate system, or just coordinates. Let us denote by
\begin{align*}
    L(z_1,\cdots, z_n) \in \mathbf{Lie},
\end{align*}
the completed free graded (super-)Lie algebra in generators $z_1, \cdots, z_n$, where each generator has possibly non-zero degree, and the completion is taken with respect to the lower central series.
Let furthermore
\begin{align*}
    \L(z_1,\cdots, z_n) \in \mathbf{LieSp},
\end{align*}
denote the Lie space whose coordinate Lie algebra is $L(z_1,\cdots, z_n)$. Thus $L(z_1,\cdots, z_n)$ and $\L(z_1,\cdots, z_n)$ are the same objects, the only difference is in the direction we choose to write morphisms, and of course in our interpretation. Using the language introduced above, the $z_1, \cdots, z_n$ are coordinates on the space $\L(z_1,\cdots, z_n)$. And elements of $\mathcal{O}(\L(z_1,\cdots, z_n))$ are Lie series in the the coordinates $z_1, \cdots, z_n$.
Let now
\begin{align*}
\L_n &:= \L (x_1, \cdots , x_n ) \\
L_n &:= L (x_1, \cdots , x_n )
\end{align*} denote the above with all generators $x_i$ of degree 0. 

In this context, $\L_n$ is nothing but the product of $n$ copies of the affine line $\L_1$, since products in $\mathbf{LieSp}$ are coproducts in $\mathbf{Lie}$ that is completed free products.
In what follows, we wish to do differential geometry on these Lie spaces. Guiding our intuition is the fact, that each element of $L_n$ induces a formal $\g$-valued function on $\g^{\times n}$. Abstractly this follows from the fact that $\g^{\times n} = \Hom(L_n, \g)$, but more concretely is it seen by just interpreting elements in $L_n$ as formulae. Take for instance $[x_1, x_2]$, it can be seen as a function taking as inputs two elements $x_1, x_2 \in \g$ and giving as output another element of $\g$.
In this sense, the space $\L_n$ can be thought of as a "universal version" of $\g^{\times n}$.
If we want to produce a $\kk$-valued function, one possibility is to take the product of two $\g$-valued functions with respect to some inner product on $\g$. Let us from now on assume that $\g$ is a quadratic Lie algebra, i.e. there is a chosen non-degenerate invariant inner product.
The definition of functions on a Lie space is then chosen such that it induces $\kk$-valued functions on $\g^{\times n}$, that is
\[
F(\L) := \O(\L) \otimes \O(\L)  / \left\{ a \otimes b - \pm b \otimes a, \ [a,b] \otimes c - a \otimes [b,c] \right\},
\]
or in other words the object in vector spaces representing the functor of symmetric invariant inner products on $L$. We will denote the universal inner product by 
\begin{align*}
    \O(\L) \otimes \O(\L) &\longrightarrow F(\L) \\
    a \otimes b &\longmapsto \left< a , b \right>.
\end{align*}

\begin{Rem}
Note that there is a difference between the space of functions and the coordinate algebra. Whereas the latter carries the structure of a Lie algebra, the former is merely a vector space, that is functions cannot be multiplied. To get an algebra, one might choose instead to work with the symmetric algebra over $F(\L)$, however, we choose not to do so.
\end{Rem}

\begin{Rem}
In terms of graphical calculus, elements of the coordinate Lie algebra can be seen as rooted Jacobi tree, whereas functions are simply Jacobi trees, where the leaves are labeled by generators of the Lie algebra. This picture will in particular explain later, why we cannot contract arbitrary forms with polyvectorfields, since this would generate loops, and thus leave the world we choose to work in.
%TODO: INSERT PICTURE!
\end{Rem}

\begin{Rem}[Ass]
As in \cite{K} everything works analogously if one replaces Lie algebras by associative algebra. Instead of developing the theory in parallel, the differences are pointed out in remarks.
In the associative world $F$ also goes under the name of $HH_0(A) = A/[A,A]$, that is the zero-th Hochschild homology. Moreover, by embedding a free Lie algebra into its universal envelopping algebra, which is a free associative algebra, all "Lie" functions embed into "Ass" functions. The last part can be seen from the Cartan-Eilenberg isomorphism $HH(U(\g)) = H_\text{Lie}(\g, (U\g)^\text{ad})$, which applied to our case says $HH_0(U(L_n)) = (U(L_n))_{L_n} \cong S(L_n)_{L_n}$, namely that "Ass" functions are the $L_n$-coinvariants of the symmetric algebra over $L_n$. In particular, we see that the quadratic part coincides with the definition of "Lie" functions.
Graphically, we are replacing Jacobi trees with ribbon trees.
\end{Rem}

In order to get forms and polyvector fields, we introduce the odd tangent and cotangent bundle, respectively,
\begin{alignat*}{2}
T[1]\left(\L(z_1,\cdots, z_n)\right) &:=  \L(z_1,\cdots, z_n, dz_1, \cdots, dz_n), \quad & |dz_i| &= |z_i|+1,  \\
T^*[1]\left(\L(z_1,\cdots, z_n)\right) &:=  \L(z_1,\cdots, z_n, \partial_1, \cdots, \partial_n), & |\partial_i| &= -|z_i|+1,
\end{alignat*}
and 
\begin{alignat*}{2}
T[1]\L_n &:=  \L(x_1,\cdots, x_n, dx_1, \cdots, dx_n), \quad & |dx_i| &= 1,  \\
T^*[1]\L_n &:=  \L(x_1,\cdots, x_n, \partial_1, \cdots, \partial_n), & |\partial_i| &= 1,
\end{alignat*}
in the non graded case.
Their functions are then denoted by
\begin{align*}
\Omega(\L)  &:= F(T[1]\L) , \\
\vf(\L) &:= F(T^*[1]\L).
\end{align*}
Both are graded vector spaces and by the usual formulae $\Omega(\L_n)$ can be endowed with a differential of degree 1. After some preparation, the usual formulae can be used to define a Lie bracket on $\vf(\L_n)$ with a Lie bracket analogous to the Schouten bracket. The Schouten bracket can be interpreted as induced by the canonical odd symplectic stucture on $T^*[1]\L_n$. It will be shown that the bracket also defines an action of polyvectorfields on the coordinate Lie algebra of $T^*[1]\L_n$. These structures are compatible with specialization, that is for any quadratic Lie algebra $\g$ we get canonical maps
\begin{eqnarray*}
\Omega(\L_n)  &\rightarrow& \Omega(\g^{\times n}) , \\
\vf(\L_n) &\rightarrow& \vf(\g^{\times n}),
\end{eqnarray*}
of complexes and Lie algebras, respectively. More concretely, let $e_\alpha$ be a basis of $\g$. Let $t_{\alpha \beta} = \left<e_\alpha, e_\beta\right>$ be the coefficients of the inner product and $t^{\alpha \beta}$ its inverse. Let $x^\alpha$ denote the dual basis of $e_\alpha$ and hence a coordinate system on $\g$. The above maps are then induced by
\begin{align*}
    \O(T[1]\L_n) &\longrightarrow \Omega(\g^{\times n}) \otimes \g \\
    x_i & \longmapsto x_i^\alpha \otimes e_\alpha \\
    dx_i &\longmapsto dx_i^\alpha \otimes e_\alpha
\end{align*}
and
\begin{align*}
    \O(T^*[1]\L_n) &\longrightarrow \vf(\g^{\times n}) \otimes \g \\
    x_i & \longmapsto x_i^\alpha \otimes e_\alpha \\
    \partial_i &\longmapsto t^{\alpha \beta} \tfrac{\partial}{\partial x_i^\alpha} \otimes e_\beta.
\end{align*}
To descend to functions, the inner product on $\g$ is applied on the $\g$ factor. The invertibility of the inner product on $\g$ is only used in the second map. A form, polyvectorfield or $\g$-valued function on $g^{\times n}$ induced by an object on $\L_n$ will be called {\it universal}. For example, the KKS Poisson bivector on $\g$, $\left< x, [\partial_x, \partial_x] \right>$, is a universal bivector field. Let us explicitly compute the image of this bivector field under the above map as follows,
\begin{align*}
\left< x^\alpha \otimes e_\alpha, \left[t^{\beta \gamma} \tfrac{\partial}{\partial x_i^\beta} \otimes e_\gamma, t^{\delta \epsilon} \tfrac{\partial}{\partial x_i^\delta} \otimes e_\epsilon\right] \right> &= t_{\alpha \eta} c_{\gamma \epsilon}^\eta t^{\beta \gamma} t^{\delta \epsilon} x^\alpha \tfrac{\partial}{\partial x^\beta} \tfrac{\partial}{\partial x^\delta} \\
&= c_{\alpha}^{\beta \delta} x^\alpha \tfrac{\partial}{\partial x^\beta} \tfrac{\partial}{\partial x^\delta},
\end{align*}
where $c_{\gamma \epsilon}^\eta$ are the structure constants of $\g$ and in the last step we raised and lowered indices using the inner product.
The adjoint action on $\g$, seen as a $\g$-valued vector field using the inner product, is universal, as it is induced by $[x,\partial_x]$.
Moreover, these objects get represented faithfully that way, as shown by
\begin{Lem}
\label{lem:faith}
Let $f \in \Omega(\L_n), \vf(\L_n)$ or $\L_n$. If $f \neq 0$ then
$f$ induces a non-zero object on $\sll(N)$ with its Killing form for $N$ sufficiently large.
\end{Lem}
\begin{proof}
Using polarization, one can reduce to the case where $f$ is linear in each coordinate. Any form or vector field that is multi-linear in the odd variables, can be seen as an ordinary multi-linear function on twice as many variables, by identifying $T\g^{\times n} \cong \g^{\times 2n}$ and $T^*\g^{\times n} \cong \g^{\times 2n}$. By embedding into the associative world, the problem is reduced to showing that the set of functions on $\sll(N)^n$ of the form
\[
\tr ( \ad_{x_{\sigma(1)}} \cdots \ad_{x_{\sigma(n)}} ) + (-1)^n \tr ( \ad_{x_{\sigma(n)}} \cdot \ad_{x_{\sigma(1)}})
\]
for$\sigma \in S_n$, are linearly independent. This can be seen by direct computation.
%TODO:make more clear!
\end{proof}

\begin{Rem}
One can also use the double of the truncated free Lie algebra (with zero cobracket) to show faithfulness.
\end{Rem}
It is clear that on any given quadratic Lie algebra we only get a comparatively small amount of functions, forms and vector fields, in particular all the objects are $\g$-invariant. The way one can use lemma \ref{lem:faith} is that whenever we have a construction or operation on a Lie space that induces a corresponding construction or operation on a concrete Lie algebra, one can use lemma \ref{lem:faith} to show that identities that hold on all lie algebras also hold on the Lie space. The theory of lie spaces can thus be thought of studying structures on $\g^{\times n}$ that are of a particularly natural type, that is in the image of the above specialization maps.

The usual yoga using contraction with the Euler vectorfield shows that $\Omega(\L_n)$ is acyclic.
Moreover, there is a simple description $\Omega^1(\L_n)$ and $\Omega^2(\L_n)$.
\begin{Lem}
\label{lem:comp}
The following maps are isomorphisms of vector spaces.
\begin{eqnarray*}
{L_n}^{\times n} & \longrightarrow & \Omega^1(\L_n) \\
(\alpha_i) & \longmapsto & \Sigma \left< dx_i, \alpha_i \right> \\
\\
\mathfrak{u} \left( n, \mathcal U (L_n) \right) = \{ (a_{ij}) \in U(L_n) \, , \, a_{ij} + *a_{ji} = 0 \}& \longrightarrow & \Omega^2(\L_n) \\
(a_{ij}) & \longmapsto & \Sigma \left< dx_i, \Ad_{a_{ij}} dx_j \right>,
\end{eqnarray*}
where $*$ is the cannonical antipode on the universal enveloping algebra of $L_n$.
\end{Lem}
In words, the lemma says that the space of 1-forms is given by an $n$-tuple of Lie series, whereas the space of 2-forms is given by a skew-symmetric matrix of associative series, where the antipode is used for the skew-symmetry.

Using the lemma we define the maps $\tfrac{\partial}{\partial x_i} : F(\L_n) \to L_n$ as the composition of $d$, the inverse of the above map, and projection onto the $i$th component, or equivalently such that 
\begin{align*}
    d\alpha = \Sigma \left< dx_i , \tfrac{\partial \alpha}{\partial x_i} \right> \quad \text{for }\alpha \in F(\L_n).
\end{align*}

\begin{Rem}[Ass]
The same is true in the "associative" world, where the role of $\mathcal U (L_n)$ is now played by $A \otimes A$, where $A$ is the underlying free associative algebra, because these objects encode functions that are linear in two additional "separator" variables.
%TODO: explain!
\end{Rem}

\begin{proof}
Surjectivity follows easily from the defining relations in both cases. For injectivity one defines the operation of contracting with the coordinate vectorfields $\partial_i$ as follows. Let $\iota_{\partial_i}$ denote the derivation of degree -1 on $\O(T[1]\L_n)$ with values in its universal enveloping algebra with module structure given my left multiplying by specifying
$$
\iota_{\partial_i} \left( dx_j \right ) = \delta_{ij}, \, \, \, \, \iota_{\partial_i} \left( x_j \right ) = 0.
$$
One checks that the following map is well-defined
\begin{align*}
    \Omega(\L_n) & \overset{\iota_{\partial_i}}{\longrightarrow} \O(T[1]\L_n) = L(x_1, \cdots, x_n, dx_1, \cdots, dx_n) \\
    \left< \alpha, \beta \right> &\longmapsto \ad^*_{\iota_{\partial_i} \alpha} \beta + (-1)^{|\alpha| |\beta|} \ad^*_{\iota_{\partial_i} \beta} \alpha.
\end{align*}
This operation has the following defining property. Let $\phi$ be the derivation of degree -1 on $\O(T[1]\L_n)$ with values in $L\left(x_1, \dotsc , x_n , dx_1, \dotsc, dx_n, t \right)$, where $t$ has degree 0, defined by
\begin{align*}
    \phi \left( dx_j \right ) = \delta_{ij} t, \, \, \, \, \phi \left( x_j \right ) = 0,
\end{align*}
which straightforwardly extends to
\begin{align*}
    \Omega(\L_n) & \overset{\phi}{\longrightarrow} F(\L(x_1, \cdots, x_n, dx_1, \cdots, dx_n, t).
\end{align*}
The maps $\iota_{\partial_i}$ and $\phi$ are then adjoint to each other in the following sense, that for any $\omega \in \Omega(\L_n)$ we have
\begin{align*}
\phi(\omega) = \left< t , \iota_{\partial_i} \left(\omega\right) \right>.
\end{align*}

Now the lemma follows by applying the operator $\iota_{\partial_i}$ to the elements of the specified form. More precisely, for forms the maps $\iota_{\partial_i}$ for $i = 1,\dots,n$ are an inverse to the map in the lemma. And in the case of two forms it follows from the observation, that any degree $1$ element in $L(x_1, \cdots, x_n, dx_1, \cdots, dx_n)$ is of the form $\sum \ad_{\pi_i} dx_i$ for uniquely determined associative series $\pi_i$.
\end{proof}

Using the above lemma we can construct the Schouten bracket. We define the following map of degree -1
\begin{align*}
\vf^\bullet  \overset{[\cdot, \cdot]}{\longrightarrow}& \Der( T^*[1]L_n) \\
\alpha \longmapsto &[\alpha,x_i] =(-1)^{|\partial_i|(|\alpha| - |\partial_i|)} \tfrac{\partial\alpha}{\partial(\partial_i)} \\
& [\alpha, \partial_i] = -(-1)^{|x_i|(|\alpha| - |x_i|)} \tfrac{\partial\alpha}{\partial x_i}.
\end{align*}
\begin{Rem}
Note that the gory signs would disappear if we choose right partial derivatives instead of left partial derivatives. The formula can then be seen as induced by
\[
[ \cdot, \cdot] = \overleftarrow{ \tfrac{\partial}{\partial (\partial_i)}} \overrightarrow{ \tfrac{\partial}{\partial x_i}} - (-1)^{|x_i| |\partial_i|} \overleftarrow{ \tfrac{\partial}{\partial x_i}} \overrightarrow{ \tfrac{\partial}{\partial (\partial_i)} }
\]
\end{Rem}

It is straightforward to check that this operation descends to a bracket on $\vf^\bullet$ and the above map intertwines this bracket with the commutator of derivations. In particular, the bracket on $\vf^\bullet$ satisfies the Jacobi identity.

Moreover, the degree one part identifies vector fields with derivations, that is
\begin{align*}
\vf^1 &\longrightarrow \Der(L_n) \\
\sum \left<\alpha_i, \partial_i \right> &\longmapsto  (x_i \mapsto \alpha_i).
\end{align*}
As the construction of $\Omega^\bullet$ is covariant, vectorfields also act on forms, we denote this action by the usual symbol for Lie derivatives, that is by $L_X$ for $X \in \vf^1$
\subsection{Poisson bivector fields}
\begin{Def}
A bivector field $\Pi \in \vf^2$ is called a {\em Poisson structure} if $[\Pi, \Pi] = 0$. It is called {\em non-degenerate} if the matrix representing it by lemma \ref{lem:comp} is not a zero-divisor.
\end{Def}
Any bivector field $\Pi$ induces a bracket by the formula $\{ \alpha, \beta \}_\Pi = [\alpha, [\Pi, \beta]]$, where $\alpha, \beta$ are either both functions, or one of the is a function and the other one an element of the coordinate algebra. If $\Pi$ is Poisson, then this defines a Lie bracket on the space of functions and an action of this Lie algebra of functions onto the underlying Lie algebra. However, the bivector field $\Pi$ carries more information than these operations, and more precisely, there exists $\Pi_1 \neq \Pi_2$ distinct bivector fields that induce the same operations, but only $\Pi_1$ is Poisson, whereas $\Pi_2$ is not (see examples of quasi-Poisson structures later).

\begin{Exa}
Any bivector field of the form 
\begin{align*}
    \left< \alpha_i, \sum_j [x_j, \partial_j] \right>
\end{align*}
for arbitrary $\alpha_i$ defines trivial operations.
\end{Exa}

\begin{Rem}
It is possible to view a bivector field as a kind of biderivation on the Lie algebra with values in its universal envelopping algebra. To make this precies, one can add two "placeholder" symbols $s$,$t$ as generators of the Lie algebra. The formula $(a,b) \mapsto \{ \left< a, s \right> , \left< b, t \right> \}_\Pi$ defines a bracket with values in functions linear in $s$ and $t$, which we can identify with the universal envelopping algebra similar as in lemma \ref{lem:comp} and in particular in its proof. Bivector fields are in one-to-one correspondence with such mappings that are biderivations in a suitable sense.
\end{Rem}

\begin{Rem}[Ass]
Using the same formula $(a,b) \mapsto \{ \left< a, s \right> , \left< b, t \right> \}_\Pi$ for $a,b \in A$ in the associative world, gives rise to a map $A \otimes A \to A \otimes A$. One checks that the theory of Poisson brackets in the associative world is equivalent to the theory of double Poisson brackets of van den Bergh \cite{vdb} in the case where the underlying algebra is free.
Moreover, the embedding of the "Lie" world into the "Ass" world preserves non-degeneracy as will be clear from the proof of lemma \ref{lem:emb} below.
\end{Rem}

\begin{Rem}
Note that the notion of non-degeneracy does not descend to specialization and is easier to satisfy.
\end{Rem}
\emptycomment{
%%%#Not true
%%%(0 x x; x 0 x^2; x -x^2 0)
\begin{Lem}
Any bivector field is the direct sum of a non-degenerate one and the zero after a linear change of coordinates.
\end{Lem}
\begin{proof}
A matrix with entries in the free associative algebra is not a zero-divisor iff it does not annihilate any vector with associative entries. If it does, then by looking at terms with a given suffix one finds that the rows are linearly dependent over the scalars. Inductively, one constructs a linear change of coordinates.
\end{proof} 
}
\begin{Exa}
The following two bivector fields can readily be seen to be non-degenerate Poisson structures.
%TODO: why non-deg?
\begin{align*}
 \Pi_\text{KKS} &= \left< \partial_x, [x,\partial_x] \right> \\
\Pi_\text{symp} &= \left < \partial_x, \partial_y \right>
\end{align*}
More examples are constructed by taking direct sums of those, which will also be denoted by the same symbol if only one type is used.
\end{Exa}
\begin{Lem}["Weinstein's splitting theorem"]
Given a Poisson structure $\Pi$, one finds coordinates, i.e. free generators $x_i,y_i,z_j$ of the underlying Lie algebra, such that $\Pi = \sum \left< \partial_{x_i}, \partial_{y_i} \right> + \tilde{\Pi}$, where $\tilde{\Pi}$ does only contain terms in the variables $z_j$ and does not contain any constant terms.
\end{Lem}
\begin{proof}
Let $\Pi_0$ denote the constant part of $\Pi$. We are trying to classify deformations of $\Pi_0$ that are governed by the dg Lie subalgebra of the Schouten Lie algebra $(\vf^\bullet, [\cdot, \cdot], d=[\Pi_0, \cdot])$ where vector fields have components at least quadratic and bivector fields are at least linear. One checks that this is actually a direct summand. One checks that $(\vf^\bullet, d)$ deformation retracts onto $(\vf(z_1, \cdots, z_n), d=0)$ where the $z_j$ form a basis of the null space of $\Pi_0$. The result now follows since the gauge group in this case is pronilpotent.

\emptycomment{
The group is a pronilpotent subgroup of the group of automorphisms of the underlying Lie algebra.

is graded by homogeneity degree. Let $\mathcal{M}$ be the degree $\geq 1$ part. It contains vector fields whose components are at least quadratic and bivector fields that are at least linear. Our deformation problem is thus governed by the dgLa $(\mathcal{M}, d=[\Pi_0, \cdot])$, whose gauge group is a the pronilpotent subgroup of the group of automorphisms of the underlying Lie algebra. 

By a linear change of coordinates we find $x_i, y_i, z_j$ such that $\Pi_0 = \sum \left< \partial_{x_i}, \partial_{y_i} \right>$ and $\tilde{\Pi}$ has no constant term, but possibly depends on $x_i,y_i$. We are looking for an automorphism of the free Lie algebra that sends $\Pi$ to the required form. The relevant cohomology is $H(\vf^\bullet, d= [\Pi_0, \cdot])$. After choosing symplectic coordinates for $\Pi_0$ one straight-forwardly defines a retraction onto the part only containing $z_j$'s.
}
\end{proof}

Each bivector field $\Pi = \Sigma \left< \partial_i, \Ad_{\Pi_{ij}} \partial_j \right> \in \vf^2$ defines a map $T^*[1]\L \overset{(-)^{\Pi}}{\rightarrow} T[1]\L$ by the formula
\begin{align*}
(dx_i)^\Pi = [\Pi, x_i] = \iota_{dx_i}\Pi = 2 \Ad_{\Pi_{ij}}\partial_j
\end{align*}

\begin{Lem}
\label{lem:emb}
Let $\Pi \in \vf^2$ be non-degenerate. Then the induced map
\begin{align*}
\Omega^n \overset{(-)^\Pi}{\longrightarrow} \vf^n
\end{align*}
is injective for $n \geq 2$.
\end{Lem}
\begin{proof}
Let $A$ denote the universal enveloping algebra of the free Lie algebra, that is the free associative algebra. Let $DA$ denote the free $A$-bimodule generated by the symbols $\partial_i$. By symmetrizing over the permutation group in n symbols the space of n-forms $\Omega^n$ can be embedded into $DA^{\otimes n} \otimes_{A\otimes A^\text{op}} A$. The map $(-)^\Pi$, for simplicity seen as an endomorphism of $\Omega^n$, extends as follows to an injective map. The non-degeneracy assumption is equivalent to $\Pi$ defining an injective map of a free right-$A$ module and consequently (using skew-adjointness of $\Pi$) also an injective map $\phi$ of a free left-$A$ module $A^m$ where $m$ is the dimension. More concretely, let $e_i$ be the canonical basis of $A^m$, then
\begin{align*}
    A^m &\overset{\phi}{\longrightarrow} A^m \\
    \alpha_i e_i &\longmapsto \alpha_i \Pi_{ij} e_j
\end{align*}
is an injective map of free left $A$modules.
Consider $A \otimes A$ as a right-$A$ module using the following $A$-bimodule structure
\begin{align*}
    (\alpha \otimes \beta).a = \alpha a^{'}\otimes *a^{''} \beta,
\end{align*}
where Sweedler's notation is used, and $*$ denotes the antipode. This defines a free $A$-module, as the invertible map $\alpha \otimes \beta \mapsto \alpha^{'} \otimes \beta^{''}$ gives a map to an obviously free $A$-module. The left-$A \otimes A^\text{op}$ coming from the outer $A$-bimodule structure is left as it is. The natural map $DA \to DA$ determined by $\Pi$ is now given by $\id \otimes \phi$, that is
\begin{alignat*}{3}
    DA &\cong (A\otimes A) \otimes_A A^m &\quad& \xrightarrow{\id \otimes \phi} \quad& (A\otimes A) \otimes_A A^m &\cong DA \\
    \alpha \partial_i \beta &\cong (\alpha \otimes \beta)  \otimes e_i &\quad& \longmapsto \quad&  (\alpha \otimes \beta) \otimes \Pi_{ij} e_j = (\alpha \Pi_{ij}^{'} \otimes *\Pi_{ij}^{''} \beta) \otimes e_j &\cong \alpha \ad_{\Pi_{ij}} (\partial_j) \beta,
\end{alignat*}
and is injective since we tensor an injective map with a free module. To finish the proof, one notices that $DA^{\otimes k}$ is a free $A$-bimodule for all $k \geq 1$ and $DA^{\otimes n} \otimes_{A\otimes A^\text{op}} A \cong DA^{\otimes n-k} \otimes_{A \otimes A^\text{op}} DA^{\otimes k}$ for any $k$ (here $DA^{\otimes n-k}$ being an $A$-bimodule is naturally a right-$A\otimes A^\text{op}$ module.
\end{proof}

\begin{Lem}
If $\Pi$ is Poisson, there is a canonical Lie bracket on $\Omega^\bullet$ such that $(-)^\Pi$ is a map of dgLAs.
\end{Lem}
\begin{proof}
By lemma \ref{lem:faith} it is enough to define the bracket on the left hand side that specializes to the classical one. To that purpose we define an odd Poisson bivector field on $T^*[1]\L_n$ as follows
\[
\widetilde{\Pi} = 2 \left< \tfrac{\partial}{\partial dx_i} \ad_{\Pi_{ij}} \tfrac{\partial}{\partial x_j} \right> +  \left< \tfrac{\partial}{\partial x_i} \ad_{d\Pi_{ij}} \tfrac{\partial}{\partial x_j} \right>,
\]
which satisfies the required properties.
\end{proof}

For later convenience we spell out the bracket on 1-forms, that is of the Lie algebra of 1-forms denoted by $\Omega^1_\Pi$. Let $\alpha = \left< \alpha_i, dx_i \right>$ $\beta = \left< \beta_j dx_j \right>$ be 1-forms, then
\begin{align*}
[ \left< \alpha_i, dx_i \right>, \left< \beta_j, dx_j \right> ]_\Pi &= \left< [ \alpha^\Pi, \beta_i ] - [\beta^\Pi, \alpha_i] , dx_i \right> + 2 \left< \alpha_i, \Ad_{d\Pi_{ij}} \beta_j \right>.
\end{align*}
The dgLA $\vf^\bullet$ is canonically filtered by (polynomial degree $-1$ + form degree $-1$). This also defines a filtration on $\Omega^\bullet_\Pi$, and by the splitting lemma it is clear that the zeroth associated graded component is a copy of a $\mathfrak{gl}_{2k}$, where $2k$ is the dimension of the symplectic vector space determined by the constant term of $\Pi$. In particular, $\Omega^1_\Pi$ is then an extension of a $\mathfrak{gl}_{2k}$ by a pronilpotent Lie algebra, and hence it's easily integrated to a group $Exp(\Omega^1_\Pi)$ together with its actions on forms and polyvector fields.
Moreover, $Exp(\Omega^1_\Pi)$ is an central extension of a Lie algebra of derivations of $\L_n$ by a factor, which in the non-degenerate case can be identified with the space of Casimir functions, that is $\{ f \in F(L_n) \ | \ [\Pi, f] = 0 \}$.
%TODO: Always split?

\begin{Rem}
The bivector field $\tilde{\Pi}$ is actually Poisson and the map $(-)^\Pi$ is a Poisson map, for a straightforward definition of Poisson maps. However, we shall have no use for this slightly stronger fact.
\end{Rem}

\begin{Rem}
%TODO: expand this!
A more geometric construction goes as follows. Let $\mathcal{C} =T^*[2]T[1]\L_n = T[1]T^*[1]\L_n$ denote the standard Courant algebroid. Let $Q \in F(\mathcal{C})$ denote the Euler vectorfield on $T[1]M$, which is a Hamiltonian for the de Rham differential on $\mathcal{C}$. The map $(-)^\Pi$ can be seen as the composition 
\[
T^*[1]\L_n \to \mathcal{C} \overset{\exp(\{\Pi, \cdot\})}{\to} \mathcal{C} \to T[1]\L_n,
\]
where the first and third maps are canonical. Thus the Courant bracet $\{ \cdot, \{ Q, \cdot \} \}$ gets twisted to $\{ \cdot, \{ \widetilde{Q}, \cdot \} \}$, where $\widetilde{Q} = \exp(\{\Pi, \cdot\})(Q)$. This is indeed a Poisson bracket on $T[1]\L_n$ iff it is at most quadratic, i.e. iff $\{ \Pi, \{ \Pi, Q \} \} = 0$, i.e. iff $\Pi$ is Poisson. In that case $\widetilde{Q} = Q + \widetilde{\Pi}$.
\end{Rem}

Maurer-Cartan elements in $\Omega^\bullet$ thus inject to the ones in $\vf^\bullet$. Namely, a Maurer-Cartan element $\sigma \in \Omega^2$ defines a new Poisson bracket $\Pi + \sigma^\Pi$. In terms of matrices (cf lemma \ref{lem:comp}) this Poisson structure is given by $\Pi - \Pi \sigma \Pi = \Pi (1 - \Pi \sigma)$. We call $\sigma \in \Omega^2$ non-degenerate if the matrix $(1 - \Pi \sigma)$ is invertible. This is in particular sufficient for $\Pi + \sigma^\Pi$ to be a non-degenerate Poisson structure. Let us denote by $\Pi_0$ and $\sigma_0$ the constant terms, then the condition is equivalent to $(1 - \Pi_0 \sigma_0)$ being non-degenerate, i.e. invertible.

\begin{Cor}
The set of non-degenerate Maurer-Cartan elements in $\Omega^\bullet$ is an $\Exp(\Omega^1_\Pi)$ homogeneous space isomorphic to $\mathcal{P} := \{\Pi+\sigma^\Pi \ | \ [\Pi+\sigma^\Pi, \Pi+\sigma^\Pi] = 0, (1 - \Pi_0 \sigma_0) \text{ invertible} \}$ where the action on the latter is by automorphism of the free Lie algebra. The stabilizer Lie algebra at $\sigma$ is isomorphic to $\Omega^0$ with Lie bracket induced by the Poisson bracket $\Pi + \sigma^\Pi$.
\end{Cor}

\begin{proof}
The isomorphism alluded to is given by $(-)^\Pi$, which is injective on forms of degree $\geq 2$. Transitivity of the $\Exp(\Omega^1_\Pi)$ action essentially follows from acyclicity of $\Omega^\bullet$. More precisely, by the splitting lemma we can write $\Exp(\Omega^1_\Pi)$ as a prounipotent extension of $\GL_k$. The action of $\Exp(\Omega^1_\Pi)$ descend to this $\GL_k$ by only considering the constant terms. It is then just the usual action of $\GL_k$ on symplectic forms on a vector space. The invertibility condition of the theorem ensures that we get a symplectic form again, that is it does not become degenerate. Let now $\sigma \in \Omega^2$ be any non-degenerate Maurer-Cartan element. By applying an element from $\GL_k$ we can assume that $\Pi + \sigma^\Pi$ has the same constant term as $\Pi$, and thus $\sigma$ lies in the pronilpotent part of $\Omega^\bullet$. Now the standard argument in an acyclic pronilpotent dgLA shows that $\sigma$ is gauge equivalent to $0$.

For the statement about the stabilizer, we only need to determine it at $\sigma = 0$ by the transitivity of the $\Exp(\Omega^1_\Pi)$ action, where it is clear.
\end{proof}

\begin{Rem}
If $\Pi$ has no constant part, then the assumption on non-degeneracy is void.
\end{Rem}

\begin{Lem}
The set $\mathcal{P}$ is in bijection with non-degenerate closed 2-forms. More precisely,
\[
\mathcal{P} = \{\Pi^\omega = (1- \Pi \omega)^{-1} \Pi \ | \ d\omega = 0, \text{ $(1- \Pi_0 \omega_0)$ invertible }\} \cong \{\omega \in \Omega^{2,cl} \ | \ \text{$(1- \Pi_0 \omega_0)$ invertible } \}.
\]
\end{Lem}
\begin{proof}
One checks that each element in $\Pi \in \mathcal{P}$ is of the form $(1- \Pi \omega)^{-1} \Pi$ for a unique $\omega \in \Omega^2$, by solving the equation of matrices with entries in the free associative algebra,
\begin{align*}
    \Pi - \Pi \sigma \Pi = (1- \Pi \omega)^{-1} \Pi,
\end{align*}
for $\omega$, namely one gets
\begin{align}
\label{eq:gaugegauge}
    \omega = \sigma (1 - \Pi \sigma)^{-1}.
\end{align}
Uniqueness follows since one uses non-degeneracy of $\Pi$. Moreover, this is well-defined, since by definition of $\mathcal{P}$ the matrix $(1 - \Pi \sigma)$ is invertible. Moreover, since $(1- \Pi \sigma) = (1 - \Pi \omega)^{-1}$, we see that the invertiblity condition is the same as in the definition of $\mathcal{P}$.
It remains to check that $d\sigma + [\sigma, \sigma]_\Pi$ is equivalent to $dw = 0$. For this purpose, let us define the map $C: \Exp(\Omega^1_\Pi) \to \Omega^2$, by the assignment $C(\phi) = \omega$ for $\phi . \Pi = \Pi^\omega$ and $\phi \in \Exp(\Omega^1_\Pi)$. One checks that this is a group 1-cocyle. Since $\Exp(\Omega^1_\Pi)$ acts transitively on $\mathcal{P}$, the image of $C$ is exactly the image of $\mathcal{P}$ under the map defined by the formula \eqref{eq:gaugegauge}. The associated Lie algebra 1-cocycle $c$ is given by $\lambda \mapsto d\lambda$, which implies that any $\omega$ defined by \eqref{eq:gaugegauge} is indeed closed. To show that we get any non-degenerate closed two-form, we first reduce to the case where $\Pi_0 \omega_0 = 0$ by using the $\GL_{2k}$ action. By solving a Moser flow type equation one shows that there is a one-parameter family $\phi_t \in \Exp(\Omega^1_\Pi)$ such that $C(\phi_t) = t \omega$. More precisely let $\omega = d\lambda$ and $\dot{\phi} \phi^{-1} = \alpha_t$, the Moser equation is then
\begin{align*}
    c(\alpha_t) + (\alpha_t).(t \omega) = \omega,
\end{align*}
with a solution given by $\alpha_t = \lambda (1 + t \Pi \omega)^{-1}$.

\emptycomment{
\begin{align*}
    \Exp(\Omega^1_\Pi) &\overset{C}{\longrightarrow} \Omega^2 \\
    \phi &\longmapsto \omega_\phi,
\end{align*}
where $

{\bf Claim 1: Each $\Pi \in \mathcal{P}$ is of the form $(1- \Pi \omega)^{-1} \Pi$ for a unique $\omega \in \Omega^2$}

One readily solves for $\omega$. Moreover, the assignment $C(\phi) = \omega$ for $\phi . \Pi = \Pi^\omega$ and $\phi \in \Exp(\Omega^1_\Pi)$ defines a group cocycle with values in $\Omega^2$. The Poisson condition now implies that the associated Lie cocycle has values in closed one-forms, which implies the result.
}
\end{proof}

\begin{Rem}
The cocycle $C: \Exp(\Omega^1_\Pi) \to \Omega^{2,cl}$ can be lifted to a group cocycle with values in $\Omega^1$ if there is a Liouville vector field for the Poisson structure, that is if there exists $X$ such that $[\Pi, X] = \Pi$. This is in particular true for any homogeneous Poisson structure. For the KKS Poisson structure this is used in \cite{ANXZ}.
\end{Rem}

\begin{Rem}
A different proof for the last steps can be obtained by showing that the formula
\begin{align*}
    [\Pi^\omega, \Pi^\omega] = 2 (d\omega)^{\Pi^\omega}
\end{align*}
holds for all non-degenerate two-forms $\omega$. One quick way of seeing this is by using lemma \ref{lem:faith} it can be reduced to the same formula in "ordinary" differential geometry, where it follows for symplectic $\Pi$ by the usual formula
\begin{align*}
    [\Pi, \Pi] = -2 (d(\Pi^{-1}))^{\Pi}
\end{align*}
and for general $\Pi$ by a density argument.
\end{Rem}

\emptycomment{
\begin{Rem}
TODO: proof formula using gauge transformation of Dirac structures.
\end{Rem}
}

\subsection{Moment Maps}
Let $\rho = \sum [x_i, \partial_i] \in T^*[1]\L_n$ denote the canonical action vector field. Note that it is independent of the choice of coordinates.
\begin{Rem}
Viewed under the natural embedding of Lie into associative algebras, $\rho$ corresponds to the cannonical element in $\Der(A, A \otimes A)$ that maps $a \mapsto 1 \otimes a - a \otimes 1$.
\end{Rem}

One defines the operator $\iota_\rho : \Omega^\bullet \to \O(T^*[1]\L_n)$ by the following procedure. Adjoin an extra variable $t$ and define $\rho^t := \left<t, \rho \right>$ as in the proof of lemma \ref{lem:comp}. Now write $\rho^t$ as $\left< \partial_i, \rho_i^t \right>$ and define a derivation $\iota_{\rho^t}$ of degree $-1$ on $T^*[1](\L_n \times \L(t)$ by sending $dx_i \mapsto \rho_i^t$, which descends to a derivation on $\Omega^\bullet(\L_n \times \L(t))$. The operator $\iota_rho$ is now defined by the formula
\begin{align*}
    \iota_{\rho^t} \alpha = \left<t, \iota_\rho(\alpha) \right> \quad \text{for }\forall \alpha \in \Omega^\bullet(\L_n)
\end{align*}

\begin{Lem}
\begin{alignat*}{2}
    \iota_{\rho^t}(d \alpha) &= [\alpha,t] && \quad \forall \alpha \in F(\L_n) \\
    \iota_\rho \left<dx_i, \Ad_{\omega_{ij}} dx_j \right> &= 2 \left<dx_i, \Ad_{x_i} \Ad_{\omega_{ij}} dx_j \right> && \quad \forall \left<dx_i, \Ad_{\omega_{ij}} dx_j \right> \in \Omega^2(\L_n) \\
\end{alignat*}
\end{Lem}

\begin{Def}
An element $\mu \in \L_n$ is called a {\em moment map} for $\Pi$ if
\begin{align*}
[\Pi, \mu] = \rho
\end{align*}
\end{Def}

\begin{Rem}
The existence of a moment map implies non-degeneracy.
\end{Rem}
\begin{Rem}
Any linear Poisson structure, for instance the linear part of a splitting of an arbitrary Poisson structure, defines a Lie algebra in the "Lie world", which by Lazard duality is a commutative algebra $C$. Non-degeneracy is equivalent to the absence of elements $\beta \in C$ such that $\alpha \beta = 0$ $\forall \alpha \in C$, and existence of a moment map is equivalent to the existence of a unit element in $C$. Moreover, the Poisson structure is rigid, if $C$ is semi-simple, i.e. does not contain nilpotents, i.e. is the product of finite field extensions. %(TODO: Check this! Identify relevant cohomology!)
%TODO: add formula! Explain!
\end{Rem}

\begin{Lem}
Let $\Pi \in \vf^2$ be a non-degenerate Poisson structure. If it admits a moment map, then it is unique.
\end{Lem}
\begin{proof}
This can be seen by rewriting the moment map condition as $d\mu^\Pi = \rho$ and using non-degeneracy of $\Pi$.
\end{proof}

The next theorem is the main result of this note, showing that Poisson structures are essentially uniquely determined by their moment maps.
Let $\Pi \in \vf^2$ be non-degenerate Poisson with moment map $\mu \in \L_n$. Let $\Pi_0$ be the constant terms of $\Pi$, in particular $\Pi_0$ defines a pairing on an n-dimensional vector space $V$. Let $Z$ denote the kernel of this pairing. One checks that the degree 2 part of $\mu$, denoted by $\mu_2$, endows $V/Z$ with a linear symplectic structure.
\begin{Thm}
\label{thm:main}
Given $\Pi \in \vf^2$ a non-degenerate Poisson structure with moment map $\mu \in \L_n$. Then we have the following.
\begin{itemize}
    \item[i)] All elements in $\mathcal{P}$ admit a unique moment map.
    \item[ii)] Assigning the moment map to a given Poisson structure in $\mathcal{P}$ defines an injective map $\mathcal{P} \to L_n$, compatible with the transitive $\Exp(\Omega^1_\Pi)$-action.
    \item[iii)] The image of this map is $\{\phi.\mu\ \ | \ \phi \in \Exp(\Omega^1_\Pi) \} = \{ \tilde{\mu} \in \mu + L_n^{\geq 2}  \ | \  \tilde{\mu}_2 \text{ is non-degenerate on $V/Z$} \}$
\end{itemize}
\end{Thm}
The theorem gives us a well-defined map
\begin{align*}
\{\phi.\mu\ \ | \ \phi \in \Exp(\Omega^1_\Pi) \} &\longrightarrow \Omega^{2, cl} \\
\eta &\longmapsto \omega^\eta,
\end{align*}
with the property that $\Pi^{\omega^\eta} = (1 - \Pi \omega^\eta)^{-1}\Pi$ is Poisson with moment map $\eta$, which is a bijection in the case where $\Pi$ has no constant terms.

\begin{proof}
\begin{itemize}
    \item [i)] Note that since each element in $\mathcal{P}$ is of the form $\phi. \Pi$ for some $\phi \in \Aut_n$ there is at least one corresponding moment map, namely $\phi.\mu$, which is unique by the previous lemma.
    \item[ii)] For a $\mu + \tilde{\mu}$ and a 2-form $\omega$, one can rewrite the equation $[\Pi^\omega, \mu + \tilde{\mu}] = \rho$ as $d\tilde{\mu} = \iota_\rho \omega$, or 
    \begin{align}
    \label{eq:mm}
    \Ad_{\tilde{\mu}_i} = - \sum_j \Ad_{x_j} \Ad_{\omega_{ji}},
    \end{align}
    where $d\tilde{\mu} = \sum \Ad_{\tilde{\mu}_i} dx_i$, $\omega = \sum \left< dx_i, \Ad_{\omega_{ij}} dx_j \right>$. This shows injectivity, namely the equation uniquely determines $\omega_{ij}$. The compatibility with the $\Exp(\Omega^1_\Pi)$-action is clear.
    \item[iii)]
    For surjectivity one can check directly that the form $\omega$ defined by equation \ref{eq:mm} is indeed closed and has the desired property. A more geometric construction suggested by $\check{S}$evera goes as follows.  Let $M$ denote the central extension of $(T[1]\L_n)^{\text{deg } \leq 1}$ by $\Omega^2$, that is $[\alpha, \beta] = \left<\alpha, \beta \right>$ for 1-forms $\alpha$, $\beta$. One checks that elements of the form $ \tilde{\mu}+ d \tilde{\mu} + \omega$ with $\omega$ closed form a subalgebra. Thus $\omega$ can be constructed by taking the degree 2 part of $\tilde{\mu}(x_1+dx_1, \dots, x_n + dx_n)$. To write down this element one needs to write $\tilde{\mu} = \sum [\mu_1^k, \mu_2^k]$ to determine
    \begin{align*}
        \omega = \left< d\mu_1^k, d\mu_2^k \right>.
    \end{align*}
    One computes
    \begin{align*}
        \iota_{\rho^t} \omega &= \left< \iota_{\rho^t} d\mu_1^k, d\mu_2^k \right> - \left< d\mu_1^k, \iota_{\rho^t}  d\mu_2^k \right>\\
        &= \left< [\mu_1^k,t], d\mu_2^k \right> - \left< d\mu_1^k, [\mu_2^k, t] \right> \\
        &= - \left< t, [\mu_1^k, d\mu_2^k] + [d\mu_1^k, \mu_2^k] \right> \\
        &= - \left< t, d \tilde{\mu} \right>.
    \end{align*}
    By computing the constant terms one checks that the condition on the mompent map is equivalent to $\omega$ being non-degenerate.
\end{itemize}
\end{proof}

\begin{Rem}[Ass]
There is also a version of the theorem in the associative case. Here the coefficients $\omega_{ij}$ lie in $A \otimes A^{op}$. Similar as in lemma \ref{lem:emb}, one can choose an automorphism of $A \otimes A$ such that $\Ad_{x_j}$ is represented by left multiplying on the left factor. Then one checks that for \ref{eq:mm} to have a solution, necessarily $\tilde{\mu} \in [A,A]$, which is also sufficient for the rest of the proof to go through.
\end{Rem}

\begin{Exa}
The theorem states, that a moment map uniquely determines a Poisson bracket in a given gauge class. However, there exist distinct Poisson brackets with the same moment map. To construct an example, consider $\Phi \in \Aut(L_3)$ given by $x_1 \mapsto x_1 + [x_2,x_3], x_2 \mapsto x_2 - [x_2, x_3], x_3 \mapsto x_3$. Clearly, $\Phi(x_1 + x_2 + x_3) = x_1 + x_2 + x_3$, however, one easily checks that $\Phi$ does not preserve $\left<\partial_i, [x_i, \partial_i] \right>$.
\end{Exa}

\begin{Rem}
In the symplectic case, the theorem is essentially equivalent to the result of Massuyeau-Turaev about non-degenerate Fox pairings (cf. \cite{MT}), which strengthens an earlier result of Kawazumi-Kuno (cf. \cite{KK}).
%TODO: Have this checked!
\end{Rem}

\renewcommand{\emptycomment}[1]{}

\subsection{Kirillov-Kostant-Souriau Poisson structure}
In this section the results are spelled out for the case of the Kirillov-Kostant-Souriau bivector field given by $\Pi := \tfrac{1}{2} \Sigma \left< x_i, [\partial_i, \partial_i] \right>$. This induces the map
$$
T^*[1]\L_n \overset{(-)^\Pi}{\longrightarrow} T[1]\L_n \, , \, dx_i = [x_i, \partial_i], \ x_i = x_i.
$$
The space of Casimir functions is determined by the following
\begin{Lem}
The kernel of the map $(-)^\Pi : \Omega(\L_n) \rightarrow \vf(\L_n)$ is linearly spanned by $\left< x_i, dx_i \right>$.
\end{Lem}

Identifying vectorfields with derivations we get a map
\begin{align*}
\Omega^1(\L_n) & \longrightarrow  \Der(\L_n) \\
\left< \alpha_i, dx_i \right> &\longmapsto (x_i \mapsto [x_i, \alpha_i]).
\end{align*}

Let us denote the Lie algebra $\Omega^1(\L_n)$ by $\tder_n$. 

The bracket on $\tder_n$ can be computed as follows. Let $\alpha = \left< \alpha_i, dx_i \right>, \beta = \left< \beta_i, dx_i \right> \in \Omega^1(\L_n)$ then
$$
[ \alpha, \beta  ] = \left< \alpha^\sharp (\beta_i) - \beta^\sharp( \alpha_i) + [\alpha_i, \beta_i] , dx_i \right>.
$$

\begin{Rem}
Our definition of $\tder_n$ differs by an $n$-dimensional abelian direct summand from the one in \cite{AT}. The same remark applies to $\TAut_n$.
\end{Rem}

Let us denote the integrating Lie group of $\tder_n$ by $\TAut_n$. The above map exponentiates to
\begin{eqnarray*}
Exp(\Omega^1(\L_n)) & \longrightarrow & \Diff(\L_n) = \Aut(L_n) \\
e^{\alpha_i dx_i} &\longmapsto& (x_i \mapsto e^{A_i} x_i e^{-A_i}) \\
&& \text{where } A_i =  \left( \frac{e^{\alpha^\sharp} - 1}{\alpha^\sharp} \right) (\alpha_i).
\end{eqnarray*}
Using the $e^{A_i}$ as components of a map, $\TAut_n$ can be thought of as $\Map(\L_n, \Exp(\L_n))$.

The closed one forms $\esder \cong \Omega^0(\L_n)$ form a Lie subalgebra, whose Lie group $\Ham(\L_n)$ can be characterized by the following lemma that appears in \cite{Dr},
\begin{Lem}[Drinfeld]
\label{drinfeld}
Let $\phi \in \eTAut$. Then
$$ 
\phi \in \Ham(\L_n)  \iff \phi(\Sigma x_i) = \Sigma x_i
$$
\end{Lem}
\begin{proof}
This follows from theorem \ref{thm:main}.
For convenience we give a direct proof.
It is enough to show that for $\alpha = \left< \alpha_i dx_i \right >$, $\alpha(\Sigma x_i) = 0$ implies that $\alpha$ is closed. We are going to use the fact that any Lie series $\alpha$ can be written as
$$
\alpha = \frac{\partial \alpha}{\partial x_i} x_i = x_i \left(\frac{\partial \alpha}{\partial x_i} \right)^* \in \kk\left<x_1, \dotsc, x_n \right>,
$$
where $*$ denotes the antipode in $\kk\left<x_1, \dotsc, x_n \right>$.
Using this we get
\begin{eqnarray*}
\alpha(\Sigma x_i) = 0 &\iff& [\alpha_i, x_i] = 0 \\
&\iff& x_i \alpha_i = \alpha_j x_j \\
&\iff& x_i \frac{\partial \alpha_i}{\partial x_j} x_j = x_i \left( \frac{\partial \alpha_j}{\partial x_i} \right)^* x_j \\
&\iff& \frac{\partial \alpha_i}{\partial x_j} = \left( \frac{\partial \alpha_j}{\partial x_i} \right)^*  \ \forall i,j \\
&\iff& d\alpha = 0
\end{eqnarray*}

\end{proof}

\begin{Lem}
\label{transitivity}
$\TAut_n$ acts transitively on $\Sigma x_i + L_n^{\geq 2}$.
\end{Lem}
\begin{proof}
Also follows from our main theorem \ref{thm:main}.
\end{proof}

\emptycomment{
\subsection{"spezialisation"}
Let now $\g$ be a fixed quadratic Lie algebra. Furthermore $\mathcal C_n$ denotes the category of $\g^n$-Hamiltonian spaces, and let $\mathcal C_n^-$ denote the category of $\g^n$-Poisson spaces. There is an obvious forgetful functor $\mathcal C \rightarrow \mathcal C^-$.

Let $\mathcal M$ denote the Maurer-Cartan elements in $\Omega(\L_n)[1]$ or equivalently the ones in $\vf(\L_n)$, namely
$$
\mathcal M := \{ \sigma \in \Omega^2(L_n) \| d \sigma + \frac{1}{2} [\sigma, \sigma] = 0 \}.
$$
Thus $\mathcal M$ consists of exactly those $\sigma \in \Omega^2(L_n)$ such that $\Pi_0 + \sigma^\sharp$ is a Poisson bracket on $\L_n$. A simpler description is given by the following
\begin{Lem}
The map
\begin{eqnarray*}
\mathcal M & \rightarrow & \Omega^{2,cl} = { \omega \ in \Omega^2 \|\  d \omega = 0} \\
\sigma & \mapsto & - \left(1 + \sigma \Pi_0 \right) \sigma
\end{eqnarray*}
is well-defined and a bijection.
\end{Lem}
\begin{proof}
\end{proof}

}

\section{Hamiltonian spaces as a {$\TAut$}-algebra}
As was shown above the groups $\TAut_n$ act transitively on $\sum x_i + L_n^{\geq 2}$. The corresponding groupoids fit together to form an operad in groupoids which we denote again by $\TAut_n$. Instead of giving the operadic compositions, a faithful (after taking a suitable limit) action on a category is constructed, from which the operadic structure can be infered.

Let $\mathcal{C}_n$ denote the category of formal $g^n$-Hamiltonian spaces, that is Poisson manifolds with a Poisson map into $g^n$ (recall that $g$ is quadratic). Using the canonical map $\mathcal{C}_n \to \mathcal{C}$, one sees that $\Fun(\mathcal{C}_n, \mathcal{C})$ form an operad in groupoids.

A map of operads $TAut_n \to \Fun(\mathcal{C}_n, \mathcal{C})$ is defined as follows.
\begin{align*}
 \sum x_i + L_n^{\geq 2} &\longrightarrow \Fun(\mathcal{C}_n, \mathcal{C}) \\
\mu & \longmapsto (M, \Pi, h) \mapsto (M, \Pi^{\omega^\mu}, \mu \circ h)
\end{align*}
where $(M, \Pi, h)$ is a Hamiltonian $g$-space with $\Pi$ its Poisson structure and $h: M \to \g^n$ its moment map. On arrows it is defined as follows.
\begin{align*}
\TAut_n &\longrightarrow \End(\Fun(\mathcal{C}_n, \mathcal{C})) \\
g : \g^n \to G^n &\longmapsto M \to M \ ; \ m \mapsto g(h(m)).m
\end{align*}
where some abuse of notation is committed and an element $g \in TAut_n$ is viewed as its induced map $\g^n \to G^n$.

\begin{Thm}
The above map is a well-defined map of operads $TAut_n \to \Fun(\mathcal{C}_n, \mathcal{C})$
\end{Thm}
\begin{proof}
The first point to notice is that $\mu \in \sum x_i + L_n^{\geq 2}$ give well-defined functors $\mathcal{C}_n \to \mathcal{C}$. Here a Poisson map $M \to \g^n$ is gauge transformed by a closed two form on $\g^n$ and then composed with a Poisson map $\mu: \g^n \to \g$.
For the second part one needs to check that a $g \in \TAut_n$ indeed intertwines the respective Poisson structures. Note that the above formula defines an action of the group $\TAut_n$ on $M$ by diffeomorphisms. It remains to check that $g$ intertwines
\begin{align*}
(M, \Pi^{\omega^\mu}, \mu \circ h) \overset{g}{\to} (M, \Pi^{\omega^{g.\mu}}, (g.\mu) \circ h).
\end{align*}
The moment map part is obvious. Moreover, the statement can be reduced to the case $\mu = \sum x_i$ by transitivity of the action. Thus the statement becomes
\begin{align*}
g. \Pi - \Pi =  \Pi^{\omega^{g.\mu}} - \Pi.
\end{align*}
Since both sides define group cocycle with values in bivector fields on $M$, it is enough to verify that the corresponding Lie cocycles coincide. Let $(u_1,\cdots, u_n) \in \tder_n$ and let moreover $\rho_i = [h_i, \Pi]$ denote the $i$-th $g$-valued action vector field on $M$, the Lie cocycle of the left hand side computes to
\begin{align*}
L_{\sum \left< u_i \circ h, \rho_i \right>}\Pi &= \sum \left< [u_i \circ h, \Pi], \rho_i \right> \\
&= \sum \left< \tfrac{\partial u_i}{\partial h_j} [h_j, \Pi], \rho_i \right> \\
&= \sum \left< \tfrac{\partial u_i}{\partial h_j} \rho_j, \rho_i \right>,
\end{align*}
which is by definition the same as the right hand side.
\end{proof}

\begin{Rem}
Not that for each $\mu \in \sum x_i + L_n^{\geq 2}$ we get a product of Hamiltonian spaces. The resulting $G$-action, however, is always the diagonal action.
\end{Rem}

\begin{Rem}
One can extend beyond formal Hamiltonian spaces by restricting to suitably convergent elements.
\end{Rem}

\emptycomment{
\begin{Thm}
Let $(M, \Pi_0)$ be a $G \times G$-Hamiltonian space with moment maps $h: M \rightarrow \g \times \g$. For any universal function $\mu: \g \times \g \rightarrow \g$ there exists a unique Poisson bivector on $M$ of the form $\Pi = \Pi_0 + \sigma_M$ for a universal 2-form on $\g \times \g$ such that $(M, \Pi)$ is a $G$-Hamiltonian space with diagonal $G$-action and moment map given by $\mu \circ h$.
Moreover, there is a group of formal diffeomorphisms of M acting transitively on such structures.
\end{Thm}

\begin{Cor}
There exists a map $\sum x_i + L_n^{\geq 2} \rightarrow \Omega(\L_n), \ \mu \mapsto \omega^\mu$ such that $(1+\Pi_0 \omega^\mu)^{-1} \Pi_0$ is a Poisson structure with moment map $\mu$. Moreover, this map is given by the de Rham differential, after suitable identification of $g$-valued 1-forms and 2-forms.
\end{Cor}
}

\subsection{Application to Hamiltonian quasi-Poisson spaces}
Let us briefly recall the relevant definitions from \cite{AM}.
Let $\phi \in \Lambda^3\g$ denote the Cartan three-form of the quadratic Lie algebra $\g$.
\begin{Def}
A pair $(M,\Pi)$ of a $\g$-manifold together with a bivector field $\Pi \in \Gamma(\Lambda^2 TM)$ is called {\it quasi-Poisson} if
$$
[\Pi, \Pi] = \phi_M,
$$
where $\phi_M$ denotes the tri-vector field on $M$ induced by $\phi$ and the $\g$-action on M.
\end{Def}
\begin{Def}
A map $\mu: M \rightarrow G$ is called a moment map if
$$
(1 \otimes \mu^* df) \Pi = (\rho \otimes df)(Z)
$$
where $Z \in \g \otimes \vf(G)$ is the adjoint action of $g$ on $G$, where $g$ and $g^*$ are identified.

A tuple $(M, \Pi, \mu)$ is called a \it{$\g$-Hamiltonian quasi-Poisson space}
\end{Def}

The category of $\g$-Hamiltonian quasi-Poisson spaces admits a monoidal structure given by the following
\begin{Def}[Fusion]
Let $(M, \Pi)$ be a  $\g \times \g$-Hamiltonian quasi-Poisson space, then
$$
\Pi_\text{fus} = \Pi - \psi_M
$$
gives a $\g$-Hamiltonian quasi-Poisson space with the diagonal $\g$-action and moment map defined by multiplying the two factors.
\end{Def}
For two $\g$-Hamiltonian quasi-Poisson spaces $M$ and $N$, we define their fusion product by
$$
M \circledast N := (M \times N, \Pi_M + \Pi_N - \psi_{M \times N}, \mu_1 \cdot \mu_2).
$$
\begin{Exa}
The moduli space of flat $\g$-connections on a surface $\Sigma_{g,n}$ of genus $g$ with $n$ boundary components, and a marked point on the boundary is given by $\Hom(\pi_1, G) \cong G^{2g+n-1}$, and carries a natural quasi-Poisson structure. It can be constructed by viewing it as $DG^{\circledast g} \circledast G^{\circledast n-1}$.
\end{Exa}
A quasi-Poisson bivector can in general be turned into a Poisson bivector by adding an r-matrix term. One particular (dynamical) r-matrix is the Alekseev-Meinrenken dynamical r-matrix. Thus the constructions goes as follows. Let $\nu(z) := \frac{1}{z} - \frac{1}{2}\coth(\frac{z}{2}) = -\tfrac{z}{12} + \tfrac{z^3}{720} + \cdots$ and define the following universal two-form on $\g$,
$$
T = \left< dx, \nu(\ad_x)dx \right>.
$$
Then recall (cf. \cite{AM})
\begin{Prop}[Exponentiation]
Let $(M, \Pi)$ be a $\g$-Poisson manifold with moment map $\mu: M \rightarrow \g$. Consider $T$ as a map $\g \rightarrow \g \wedge \g$. Then
$(M, \Pi - (\mu^*T)_M)$ is a $\g$-Hamiltonian quasi-Poisson manifold with moment map $\exp \circ \mu$.

The bivector field can also be written as $\Pi - (\mu^*T)^\sharp$ by considering $T$ as a two-form on $\g$ and using the morphism $\sharp$ between forms and polyvector fields induced by $\Pi$.
\end{Prop}

Let $\Exp$ denote the functor sending a $\g$-Hamiltonian Poisson space to the $\g$-Hamiltonian quasi-Poisson space given by the last proposition.
\begin{Exa}
The standard $\g$-Hamiltonian quasi-Poisson space $G$ with moment map the identity, corresponds to $\g$ with its KKS structure under this functor.
\end{Exa}

The theorem is equivalent to $T \in \Omega^2(\L_1)$ satisfying the dynamical Yang-Baxter equation
\[
-2 dT + [T,T]_{\Pi_{KKS}} = \tfrac{1}{6} \left< dx, [dx,dx] \right>.
\]

Pulling back the fusion product along the functor $\Exp$ one gets a second monoidal structure on the category of $\g$-Hamiltonian spaces, which we denote again by $\circledast$. It is given by
\[
(M,\Pi_M, \mu_M) \circledast (N,\Pi_N, \mu_N) = (M \times N, \Pi_M + \Pi_N + ((\mu_M \times \mu_N)^*\sigma)^\sharp, \log(e^{\mu_M} e^{\mu_N}))
\]
for
$$
\sigma = T_{12} - T_1 - T_2 +  \left<dx, dy \right> \in \Omega^2(\L_n).
$$
This $\sigma$ is thus a Maurer-Cartan element in $\Omega^2(\L_n)$ since it is so for any $\sll_k$. Setting $\omega = \sigma (1 + \Pi \sigma)^{-1} \in \Omega^2{\L_n}$, the above Poisson structure can be written as $(\Pi_M \times \Pi_N)^{(\mu_M \times \mu_N)^*\omega)}$ and $\omega = \omega^{\log(e^{x_1} e^{x_2})}$.
In particular, there two products are induced by $x_1 + x_2$ and $\log(e^{x_1} e^{x_2})$, respectively.

Let now $F \in \TAut_2$ be an element intertwining those two structures, that is such that 
\[
\label{eq:special}
 F(\log(e^{x_1} e^{x_2})) = x_1 + x_2
\]
\begin{Rem}
As shown in \cite{AET} one particular source of such $F$ is Drinfeld associators. Namely, let $\Phi = \exp(\phi)$ for $\phi \in \hat{Lie}(x,y)$ be a Drinfeld associator. Then we associate to it the $F_\Phi$ with components $$\left(\Phi(x,-x-y), e^{-\frac{x+y}{2}} \Phi(y,-x-y) \right).$$ 
\emptycomment{Then the following picture shows that $F_\Phi$ intertwines the Baker-Campbell-Hausdorff formula and the sum as in \eqref{eq:special}.
---insert picture----}
\end{Rem}
As a consequence of the above discussion, we get the following.
\begin{Prop}
Let $M,N$ be two $g$-Hamiltonian Poisson spaces. Then the following map is Poisson.
\begin{eqnarray*}
M \times N &\overset{F_{M,N}}{\rightarrow}& \Exp^{-1}\left(\Exp(M) \circledast \Exp(N) \right) \\
(a,b) &\mapsto& (F_1(\mu_M(a), \mu_N(b)). a, F_2(\mu_M(a), \mu_N(b)).b)
\end{eqnarray*}
Moreover, the $F_{M,N}$ are a natural transformation.
\end{Prop}
\emptycomment{
\begin{Prop}
The map $F_{M,N}$ is a Poisson map.
\end{Prop}
\begin{proof}
The Poisson bivector on  $\Exp^{-1}\left(\Exp(M) \circledast \Exp(N) \right)$ is given by
$$
\Pi_M + \Pi_N - (\mu_M^* T)^\sharp - (\mu_N^*T)^\sharp + \psi_{M \times N} +( (\log(e^{\mu_M} e^{\mu_N}))^*T)^\sharp,
$$
which can be written as
$$
\Pi_M + \Pi_N + ((\mu_M \times \mu_N)^*\sigma)^\sharp
$$
for
$$
\sigma = T_{12} - T_1 - T_2 +  \left<dx, dy \right> \in \Omega^2(\L_n)
$$
where $T_i$ and $T_{12}$ are defined via the pullback of $T = \left< dx, \nu(\ad_x)dx \right> \in \Omega^2(\L_1)$ along the maps $x_1$, $x_2$ and $\log(e^{x_1}e^{x_2})$, respectively.
In particular we see that the Poisson structure is of the required form and it is thus enought to respect the moment maps, which $F$ does by assumption.
\end{proof}
}
The maps $F_{M,N}$ can now be interpreted as a monoidal structure on the functor $\Exp$. Let $\mathcal{C}$ denote the category of $g$-Hamiltonian Poisson spaces with monoidal product given by the product of Poisson spaces. Instead of the trivial associator isomorphism, let $\mathcal{C}$ be endowed with the associator derived from $F$. More precisely, define
$$
\Phi^F = F^{}_{1,23} F^{}_{2,3} F_{1,2}^{-1} F_{12,3}^{-1} \in \TAut_3,
$$
and use it to define a diffeomorphism $\Phi^F_{X,Y,Z}$ for any triple $X,Y,Z$ of $\g$-Hamiltonian Poisson spaces. Let $\mathcal{D}$ denote the category of $g$-Hamiltonian quasi-Poisson spaces.
Then we get
\begin{Prop}
An $F \in \TAut_2$ such that $F(\log(e^{x_1} e^{x_2})) = x_1 + x_2$ promotes the functor $\Exp$ to a monoidal equivalence
\begin{eqnarray*}
(\mathcal{C}, \times, \Phi^F) &\overset{\Exp}{\longrightarrow}& (\mathcal{D}, \circledast, \id)
\end{eqnarray*}
\end{Prop}

\newcommand{\spq}{\big{/}\!\! \big{/}_{\! \! 0} \, }
\begin{Cor}
An $F \in \eTAut[2]$ satisfying \eqref{eq:special} gives a Poisson map
$$
O_{\lambda_1} \times \cdots \times O_{\lambda_n} \spq G \rightarrow \mathcal{M}(\Sigma_{0,n}, C_1, \cdots, C_n)
$$
where $O_{\lambda_i}$ are coadjoint orbits for given $\lambda_i \in \g \cong \g^*$, and $\mathcal{M}(\Sigma_{0,n}, C_1, \cdots, C_n)$ is the moduli space of flat connections on a surface of genus 0 with $n$ punctures and monodromies around the punctures prescribed by conjugacy classes $C_i = G.exp(\lambda_i)$.
\end{Cor}

\begin{Rem}
Taking $F = F_{\Phi_{KZ}}$ to be associated to the Knizhnik-Zamolodchikov associator, the previous map is given by
$$
a_1, a_2 \mapsto d-\left(\frac{a_1}{z} + \frac{a_2}{z-1} \right) dz
$$
\end{Rem}

\end{document}